\documentclass[runningheads]{llncs}

\usepackage[T1]{fontenc}
\usepackage{graphicx}
\usepackage{times}
\usepackage{soul}
\usepackage{url}
\usepackage[hidelinks]{hyperref}
\usepackage{inputenc}
\usepackage{amsmath}
\usepackage{amssymb}
\usepackage{amsfonts}
\usepackage{latexsym}
\usepackage{booktabs}
\usepackage{complexity}
\usepackage[Algorithm,ruled]{algorithm}
\usepackage[noend]{algpseudocode}
\usepackage{orcidlink}
  
\usepackage{amsmath}
\usepackage{booktabs}
\usepackage{latexsym} 
\usepackage{amssymb}
\usepackage{complexity}
\usepackage{multirow}
\usepackage{verbatim}
\usepackage{stmaryrd}
\usepackage{multirow}
\usepackage[group-separator={,}]{siunitx}

\usepackage[Algorithm,ruled]{algorithm}
\usepackage[noend]{algpseudocode}

\newtheorem{mytheorem}{\bf Theorem}

\newtheorem{mydefinition}{\bf Definition}
\newtheorem{myexample}{\bf Example}

\newcolumntype{C}{>{\centering\arraybackslash}p{1cm}}
\newcommand*{\QEDB}{\null\nobreak\hfill\ensuremath{\blacksquare}}

\begin{document}

\title{Fair Division meets Vehicle Routing: \\ Fairness for Drivers with Monotone Profits}
\titlerunning{Fair Division meets Vehicle Routing}

\author{Martin Damyanov Aleksandrov\orcidlink{0000-0003-0047-1235}}
\authorrunning{Martin D.\ Aleksandrov}
\institute{Freie Universit\"{a}t Berlin, Germany \\
        {\tt\small martin.aleksandrov@fu-berlin.de}}

\maketitle

\begin{abstract}
We propose a new model for fair division and vehicle routing, where drivers have monotone profit preferences, and their vehicles have feasibility constraints, for customer requests. For this model, we design two new axiomatic notions for fairness for drivers: FEQ1 and FEF1. FEQ1 encodes driver pairwise bounded equitability. FEF1 encodes driver pairwise bounded envy freeness. We compare FEQ1 and FEF1 with popular fair division notions such as EQ1 and EF1. We also give algorithms for guaranteeing FEQ1 and FEF1, respectively.
\end{abstract}

\section{Introduction}\label{sec:intro}

Let us consider the classical Vehicle Routing Problem (VRP) \cite{dantzig1959}. In this problem, there is a single vehicle and a set of visit requests. Generalisations of the VRP consider a fleet of multiple vehicles and a set of pickup-and-delivery requests, see e.g.\ \cite{savelsbergh1995}. We study social aspects in VRPs. We thus define a domain for {\em Social VRPs (SVRPs)}. In Figure~\ref{fig:svrps}, we depict this domain.

\begin{figure}[h]
\centering
\includegraphics[width=0.75\columnwidth,height=6cm]{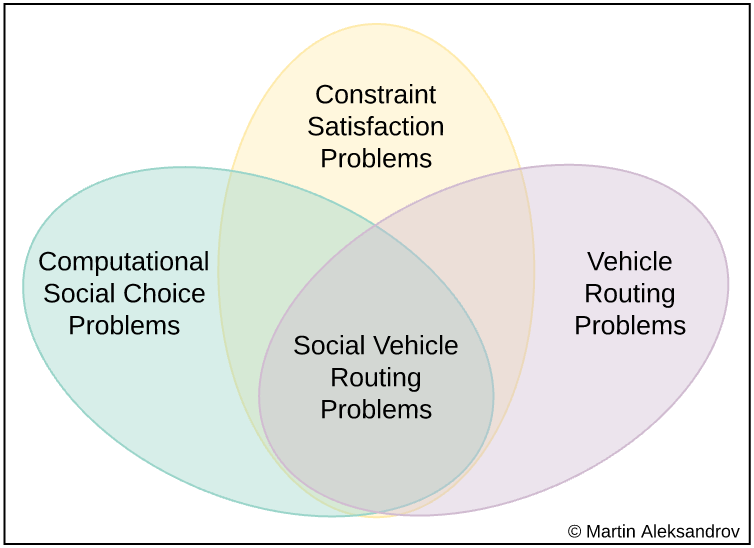}
\caption{The Social Vehicle Routing Problem.}
\label{fig:svrps}
\end{figure}  

SVRPs can combine features from fair division, voting theory, judgment aggregation, planning, loading, and search. In this paper, as an initial step, we draw inspiration from monotone fair division and VRPs. Thus, we add to the initial SVRP research from \cite{aleksandrov2021iv,aleksandrov2021epia,aleksandrov2021hfif,aleksandrov2021optima}. We have every confidence that this new domain lays down the blueprints for much further research in the future. 

Our inspiration is in line with multiple programmes across Europe in regard to emerging VRPs. Emerging VRPs include autonomous vehicles, connected vehicles, electric vehicles, intelligent garbage vehicles, and data-driven logistics. For example, the 2020 EU Strategy for Sustainable and Smart Mobility has formulated public mobility Transport Policy Flagships, according to which the transition to emerging VRPs must involve the {\em preferences} of individuals. 

In the previous years, there are some preliminary steps in the direction of this transition: in 2016, the German Ministry for Traffic and Digital Infrastructure has granted one hundred million Euros for autonomous and connected driving; in 2017, the German Ministry appointed Ethics Commission to regulate the use of such vehicles for social good; in 2020, the European Institute of Innovation and Technology in Hungary received four hundred million Euros from EU for encouraging people to use electric vehicles more frequently in order to make our cities more livable places. Furthermore, to continue the discussion about the transition to future public mobility, the 2022 Transport Research Arena hosts a number of tracks on this topic.

Besides achieving the safety and explainability of emerging VRP technologies, another goal in all these programmes is achieving trust. By Deloitte's Trustworthy AI Framework, promoting trust requires that vehicles are used fairly. In this paper, we provide an early qualitative analysis of {\em fairness for drivers}. Achieving fairness for drivers might be orthogonal to a common VRP objective such as minimising the total time travelled by all vehicles. We illustrate this incompatibility in Example~\ref{exp:one}.

\begin{myexample}\label{exp:one}
{\bf (Fairness for drivers)} Let us consider $[0,1]$. Also, let there be vehicle $v_1$ at $0$, vehicle $v_2$ at $1$, visit request $r_1$ at $\epsilon\in [0,\frac{1}{4})$, and visit request $r_2$ at $2\cdot\epsilon$. Suppose that each driver derives a profit of $1$\$ per visit. This could be the cost they charge the customer who submitted the request. Minimising the total time travelled by $v_1$ and $v_2$ dispatches first $v_1$ to $\epsilon$ and then to $2\cdot\epsilon$. The objective value is equal to $2\cdot\epsilon$. However, this outcome is {\em not fair} for drivers because the driver of $v_1$ receives a profit of $2$\$ whereas the driver of $v_2$ receives a profit of $0$\$. By comparison, assigning $r_1$ to $v_1$ and $r_2$ to $v_2$, and then dispatching $v_1$ to $\epsilon$ and $v_2$ to $2\cdot\epsilon$, gives an objective value of $(1-\epsilon)$ for the total time travelled by both vehicles. For $\epsilon\in [0,\frac{1}{4})$, this value is strictly greater than $2\cdot\epsilon$. Although this outcome is sub-optimal for the total travel time objective, it is {\em fair} for drivers because each of them receives a profit of $1$\$.\QEDB
\end{myexample}

In Example~\ref{exp:one}, we also note that achieving fairness for drivers requires us to {\em isolate} the assignment sub-problem from the routing sub-problem. In this paper, we do exactly this. As a first step, we formalise a fair division model where fleet $V=\lbrace v_1,\ldots,v_n\rbrace$ of $n\in\mathbb{N}_{\geq 0}$ vehicles services set $R=\lbrace r_1,\ldots,r_m\rbrace$ of $m\in\mathbb{N}_{\geq 0}$ requests. Thus, within this formalisation, we will show that solving the assignment sub-problem can {\em always} guarantee fairness for drivers. 

\section{Profits and feasibilities}\label{sec:pf}

In our model, for the driver $i$ of each $v_i$, we consider a {\em monotone} profit function $p_i:2^R\rightarrow\mathbb{R}_{\geq 0}$ for the requests from $R$: for each $S\subseteq R$ and each $r_j\in R$, $p_i(S\cup\lbrace r_j\rbrace)\geq p_i(S)$. We let $p_i(\emptyset)=0$ and $p_{ij}=p_i(\lbrace r_j\rbrace)$. Three common types of such functions are additive, sub-additive, and super-additive monotone functions \cite{mirchandani2013}. For $S\subseteq R,T\subseteq R$, additivity of $p_i$ requires $p_i(S\cup T)=p_i(S)+p_i(T)$, sub-additivity of $p_i$ requires $p_i(S\cup T)\leq p_i(S)+p_i(T)$, and super-additivity of $p_i$ requires $p_i(S\cup T)\geq p_i(S)+p_i(T)$. We consider the more general monotone profits.

In our model, each vehicle can be feasible or not for every request. We encode thus {\em feasibilities} of vehicles for requests. In detail, we suppose that there is some set of feasibility constraints $C_{ij}$ for each $v_i$ and each $r_j$. We model thus the feasibility of $v_i$ for $r_j$ by using a binary indicator $f_{ij}$ such that $f_{ij}=1$ if all constraints in $C_{ij}$ can be satisfied, and else $f_{ij}=0$. We can use these indicators to model various practical constraints. 1-D vehicle capacities and 3-D package dimensions are two common types of such constraints \cite{mannel2018}.

\section{Contributions}\label{sec:cont}

We consider assignments of requests to drivers. An {\em assignment} is $(R_1,\ldots,R_n)$, where $R_i\subseteq R$ for each $v_i$ and $R_i\cap R_j=\emptyset$ for each $(v_i,v_j)$ such that $i\neq j$. We say that it is {\em feasible} if, for each $v_i$ and each $r_j\in R_i$, $f_{ij}>0$ holds. We say that it is {\em complete} if, for each $r_j\in R$ with $f_{ij}>0$ for some $v_i\in V$, $r_j\in R_k$ holds for some $v_k\in V$, and for each $r_j\in R$ with $f_{ij}=0$ for each $v_i\in V$, $r_j\not\in\textstyle\bigcup_{v_i\in V} R_i$ holds. In our work, we focus on combining feasible complete assignments with fairness for drivers. In Section~\ref{sec:rel}, we review some related works. In Section~\ref{sec:feasone}, we consider in our context two popular such notions: EF1 and EQ1. For any two drivers, EQ1/EF1 assignments bound the absolute {\em jealousy} (i.e.\ objective profit difference)/{\em envy} (i.e.\ subjective profit difference) for the requests for which their vehicles are feasible or {\em not}. EQ1 and EF1 were proposed in fair division of goods, where feasibilities for goods are absent. For this reason, they ignore the feasibilities in our model. This leads to an impossibility result for feasible complete EQ1/EF1 assignments: Theorem~\ref{thm:one}. We then prove that deciding whether such assignments exist is $\NP$-hard: Theorem~\ref{thm:two}. In Section~\ref{sec:feastwo}, we propose respectively two new notions: FEQ1 and FEF1. For any two drivers, FEQ1/FEF1 assignments bound the absolute {\em feasible jealousy/envy} between them, i.e.\ the jealousy/envy for the requests for which their vehicles are {\em feasible}. In settings where each vehicle is feasible for each request (i.e.\ with {\em unit} feasibilities), FEQ1 coincides with EQ1 and FEF1 coincides with EF1. We prove that feasible complete assignments may falsify FEQ1/FEF1 and complete FEQ1/FEF1 assignments may falsify feasibility: Theorem~\ref{thm:three}. Nevertheless, we give a polynomial-time algorithm for computing feasible complete FEQ1 assignments: Theorem~\ref{thm:four}. Furthermore, we give another pseudo-polynomial-time algorithm for computing FEF1 assignments as well: Theorem~\ref{thm:five}. Figure~\ref{fig:two} depicts our results.

\begin{figure}[t]
\centering
\includegraphics[width=0.45\columnwidth,height=4cm]{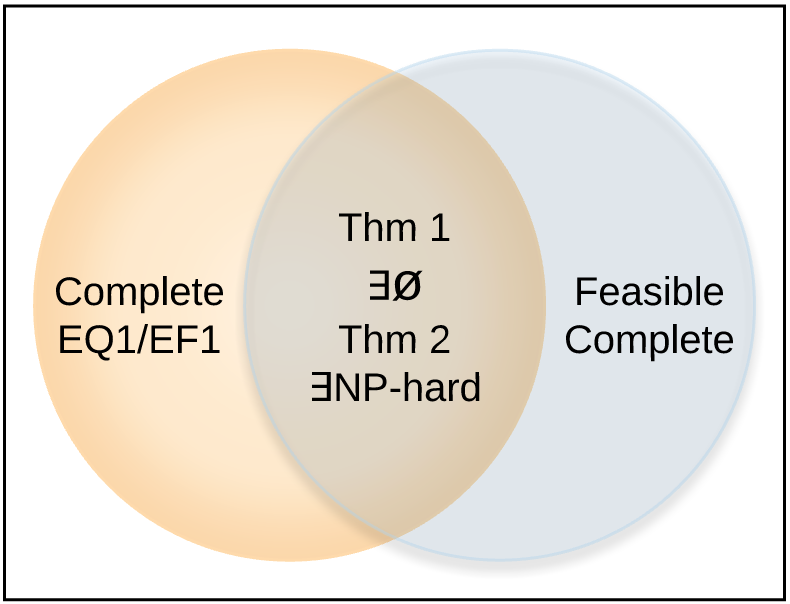}\hspace{0.5cm}\includegraphics[width=0.45\columnwidth,height=4cm]{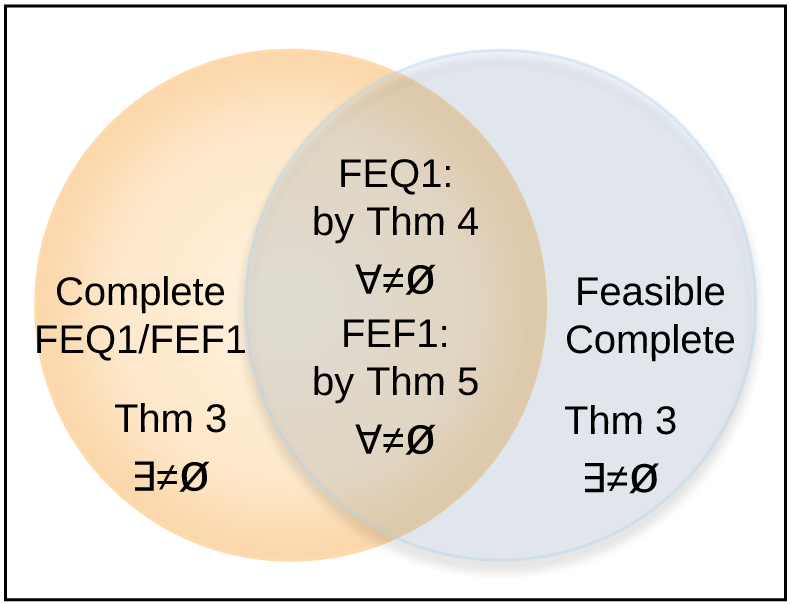}
\caption{An overview of our results for fairness for drivers with monotone profits and vehicle feasibilities: $\exists\emptyset$ means that there is a setting with additive profits, where an assignment does not exist (Theorem~\ref{thm:one}); $\exists\NP$-hard means that there is a setting with additive profits, where deciding whether an assignment exists takes exponential time (Theorem~\ref{thm:two}); $\exists\neq\emptyset$ means that there is a setting with additive profits, where an assignment exists (Theorem~\ref{thm:three}); $\forall\neq\emptyset$ means that there is no setting with monotone profits, where an assignment does not exist (by Theorems~\ref{thm:four} and~\ref{thm:five}, respectively).}
\label{fig:two}
\end{figure}  

\section{Some related work}\label{sec:rel}

Santos and Xavier \cite{douglas2013} studied minimising shared taxi costs among customers. Li et al.\ \cite{baoxiang2014} considered a similar setting with people and parcels. Another work about customer ride-sharing is \cite{bei2018}. In this paper, we consider {\em driver} (i.e.\ not customer) preferences. Ma, Fang, and Parkes \cite{hongyao2019} proposed a mechanism for spatio-temporal settings with additive pricing for drivers, where envy-freeness for any fixed two drivers is guaranteed only if they are at the same (i.e.\ not different) locations. Rheingans-Yoo et al.\ \cite{rheingans2019} analysed matchings with driver location (i.e.\ neither profit nor feasibility) preferences. Xu and Xu \cite{yifan2020} investigated trading the system efficiency for the income equality of drivers. Lesmana, Zhang, and Bei \cite{lesmana2019} studied fairness for drivers in terms of maximising the minimum additive (i.e.\ not monotone) trip profit of any driver. We add to these works by analysing new {\em bounded feasible envy-freeness} (i.e.\ FEF1) and new {\em bounded feasible equitability} (i.e.\ FEQ1) in various practical domains with {\em monotone} profits and {\em feasibility} constraints. The bound is the greatest marginal profit of a request.

\section{Feasible complete EQ1/EF1 assignments}\label{sec:feasone}

We begin our analysis with interactions between feasible complete assignments and popular fairness concepts such as EQ1 and EF1. For this purpose, we define EQ1 and EF1 in our terms.

\begin{mydefinition}
$(R_1,\ldots,R_n)$ is {\em EQ1} if, for each $v_i,v_k\in V$ s.t.\ $R_k\neq\emptyset$, $p_i(R_i)\geq p_k(R_k\setminus\lbrace r_j\rbrace)$ holds for some $r_j\in R_k$.
\end{mydefinition}  

\begin{mydefinition}
$(R_1,\ldots,R_n)$ is {\em EF1} if, for each $v_i,v_k\in V$ s.t.\ $R_k\neq\emptyset$, $p_i(R_i)\geq p_i(R_k\setminus\lbrace r_j\rbrace)$ holds for some $r_j\in R_k$.
\end{mydefinition}  

There are settings with unit feasibilities where the set of complete EQ1 assignments and the set of complete EF1 assignments are disjoint. We show this incompatibility in Example~\ref{exp:two}.

\begin{myexample}\label{exp:two}
Let us consider $v_1$, $v_2$ and visit $r_1,\ldots,r_4$. Define the profits as: for each $r_j$, $p_{1j}=4$ and $p_{2j}=1$. Suppose additive profits and unit feasibilities. We make two observations. To achieve completeness and EF1, we must give exactly two requests to each driver. Otherwise, one of the drivers would envy the other driver, even after removing one request from the other driver's bundle. 

However, each complete EF1 assignment violates EQ1 because driver $2$ has a profit of $2$ which is strictly lower than the profit of $4$ of driver $1$ after a request from the bundle of driver $1$ is removed. In contrast, each complete EQ1 assignment gives one request to driver $1$ and three requests to driver $2$. \QEDB
\end{myexample}

In fair division without feasibilities, complete EQ1 assignments can be computed in polynomial time, and complete EF1 assignments can also be computed in polynomial time \cite{lipton2004}. 

By comparison, in our setting with feasibilities, each feasible complete assignment may fail to satisfy EQ1/EF1 in some settings. The problem here is that such an assignment may give all requests to a given driver.

\begin{mytheorem}\label{thm:one}
In some settings with vehicle feasibilities and additive driver profits, all complete EQ1/EF1 assignments may not be feasible and all feasible complete assignments may not be EQ1/EF1.
\end{mytheorem}

\begin{proof}
Let there be $v_1,v_2$, and $r_1,r_2$. Also, define any strictly positive profits, and $f_{11}=f_{12}=1$ but $f_{21}=f_{22}=0$. For example, the profits could be all equal to $1$. Suppose drivers have additive profits. We make the following two observations for this setting. 

First, to achieve completeness and EQ1/EF1, we must assign exactly one request to each vehicle. There are two drivers and two requests, so there are two complete EQ1/EF1 assignments. Hence, each such assignment gives some infeasible request to $v_2$. Such an outcome cannot be feasible. 

Second, the unique feasible complete assignment gives both requests to $v_1$. This outcome violates EQ1/EF1 because the profits are strictly positive, and removing any single request from the bundle of the driver of $v_1$ does not eliminate the jealousy/envy of the driver of $v_2$. The result follows. \QEDB
\end{proof}

In some settings with feasibilities, we conclude that the set of feasible complete EQ1/EF1 assignments could be empty. In response, we study the computational complexity of deciding whether feasible complete EQ1/EF1 assignments exist. We formulate therefore the following decision problem and prove that it is weakly $\NP$-hard.

\begin{center}
\fbox{
\begin{minipage}{0.95\columnwidth}
Problem: {\sc Feasible Complete EQ1/EF1} \\
{\bf Data}: $V$, $R$, $(p_{ij})_{n\times m}$, $(f_{ij})_{n\times m}$. \\
{\bf Query}: Is there feasible complete EQ1/EF1 assignment $(R_1,\ldots,R_n)$?
\end{minipage}
}
\end{center}

\begin{mytheorem}\label{thm:two}
In some settings with vehicle feasibilities and additive driver profits, {\sc Feasible Complete EQ1/EF1} belongs to the class of weakly $\NP$-hard decision problems.
\end{mytheorem}  

\begin{proof}
We reduce from the popular weakly $\NP$-hard partition problem. For multiset $\lbrace p_1, \ldots , p_m\rbrace$ of positive numbers, whose sum is equal to $2\cdot P>m\geq 1$, this problem asks whether there is an equal partition of $\lbrace p_1, \ldots , p_m\rbrace$ into two multisets, where the sum of each multiset is $P$. Given an instance of the partition problem, let us construct an instance of {\sc Feasible Complete EQ1/EF1} with $3$ drivers and $(m+4)$ requests: $m$ number-requests $\lbrace r_1, \ldots , r_m\rbrace$ and four special-requests $\lbrace s_{1}, \ldots, s_{4}\rbrace$. 

\begin{center}
\begin{tabular}{c|cCCCC}
\hline
Requests & $\forall j\in\lbrace 1,\ldots,m\rbrace: r_j$ & $s_{1}$ & $s_{2}$ & $s_{3}$ & $s_{4}$  \\ \hline
\multicolumn{6}{c}{profits}\\
\hline
Driver $1$ & $\forall j\in\lbrace 1,\ldots,m\rbrace: p_j$ & $2\cdot P$ & $2\cdot P$ & $3\cdot P$ & $3\cdot P$  \\
Driver $2$ & $\forall j\in\lbrace 1,\ldots,m\rbrace: p_j$ & $2\cdot P$ & $2\cdot P$ & $3\cdot P$ & $3\cdot P$  \\
Driver $3$ & $\forall j\in\lbrace 1,\ldots,m\rbrace: p_j$ & $2\cdot P$ & $2\cdot P$ & $3\cdot P$ & $3\cdot P$ \\
\hline
\multicolumn{6}{c}{feasibilities}\\
\hline
Driver $1$ & $\forall j\in\lbrace 1,\ldots,m\rbrace: 0$ & $0$ & $0$ & $1$ & $1$  \\
Driver $2$ & $\forall j\in\lbrace 1,\ldots,m\rbrace: 1$ & $1$ & $0$ & $0$ & $0$  \\
Driver $3$ & $\forall j\in\lbrace 1,\ldots,m\rbrace: 1$ & $0$ & $1$ & $0$ & $0$ \\
\hline
\end{tabular}
\end{center}

The profits sum up to $12\cdot P$. Also, we note that a feasible complete assignment cannot give $r_j$ to driver $1$ because their vehicle is not feasible for $r_j$, i.e.\ $f_{1j}=0$. Thus, an assignment is feasible complete iff $r_1,\ldots,r_m$ go to drivers $2$ and $3$, $s_{1}$ goes to driver $2$, $s_{2}$ goes to driver $3$, and $s_{3},s_{4}$ go to driver $1$. In this instance, we note that an assignment is EQ1 iff it is EF1. This follows by the fact that all three drivers have additive profits and, for each request, all of them have equal profits.

Let there be an equal partition of the integers. Then, we can assign to each of the drivers $2$ and $3$ a profit of $P$ from $r_1,\ldots,r_m$ plus a profit of $2\cdot P$ from $s_{1}$ and $s_{2}$, for a total profit of $3\cdot P$. Also, we can assign to driver $1$ a profit of $6\cdot P$ from $s_{3}$ and $s_{4}$. This assignment is feasible complete. Furthermore, driver $1$ with a profit of $6\cdot P$ is EQ1/EF1 of driver $2$ or driver $3$ because each of the profits of drivers $2$ and $3$ for their own bundles is $3\cdot P$, and the profit of driver $1$ for each of their bundles is also $3\cdot P$. Furthermore, drivers $2$ and $3$ are also EQ1/EF1 of driver $1$ once $s_{3}$ or $s_{4}$ is removed from the bundle of driver $1$, for an equal EQ1/EF1 profit of $3\cdot P$. Hence, the assignment is EQ1/EF1. 

Let there be a feasible complete EQ1/EF1 assignment. Hence, driver $1$ gets only $s_{3}$ and $s_{4}$. As a result, drivers $2$ and $3$ have each a profit of $6\cdot P$ for the bundle of driver $1$. Once $s_{3}$ or $s_{4}$ from this bundle is removed, drivers $2$ and $3$ have each a profit of $3\cdot P$ for the remaining item. Hence, each of them must get a profit of at least $3\cdot P$ from their own bundle. As driver $2$ gets $s_{1}$ and driver $3$ gets $s_{2}$, so each of them must get a profit of exactly $P$ from $r_1,\ldots,r_m$. There must be an equal partition of the integers. \QEDB
\end{proof} 

Finally, unlike complete EQ1/EF1 assignments which exist and are tractable, feasible complete EQ1/EF1 assignments may not exist, and deciding whether such assignments exist is intractable.

\section{Feasible complete FEQ1/FEF1 assignments}\label{sec:feastwo}

We continue our analysis with interactions between feasible complete assignments and the new fairness concepts FEQ1 and FEF1. For this purpose, we define FEQ1 and FEF1.

\begin{mydefinition}
$(R_1,\ldots,R_n)$ is {\em FEQ1} if, for each $v_i,v_k\in V$ where $F_{ii}=\lbrace r_j\in R_i | f_{ij}>0\rbrace$ and $F_{ik}=\lbrace r_j\in R_k | f_{ij}>0\rbrace\neq\emptyset$, $p_i(F_{ii})\geq p_k(F_{ik}\setminus\lbrace r_j\rbrace)$ holds for some $r_j\in F_{ik}$.
\end{mydefinition}

\begin{mydefinition}
$(R_1,\ldots,R_n)$ is {\em FEF1} if, for each $v_i,v_k\in V$ where $F_{ii}=\lbrace r_j\in R_i | f_{ij}>0\rbrace$ and $F_{ik}=\lbrace r_j\in R_k | f_{ij}>0\rbrace\neq\emptyset$, $p_i(F_{ii})\geq p_i(F_{ik}\setminus\lbrace r_j\rbrace)$ holds for some $r_j\in F_{ik}$.
\end{mydefinition}

These two notions may differ: see Example~\ref{exp:two}. That is, there are settings where the set of feasible complete FEQ1 assignments and the set of feasible complete FEF1 assignments are disjoint. 

Furthermore, a complete FEQ1/FEF1 assignment may not be feasible in some settings. This is possible because the profit of a given driver for a given feasible request could be large enough to guarantee these properties. Also, a feasible complete assignment may not be FEQ1/FEF1 because it may give all requests to a given driver.

\begin{mytheorem}\label{thm:three}
In some settings with vehicle feasibilities and additive driver profits, some complete FEQ1/\-FEF1 assignments may not be feasible, and some feasible complete assignments may not be FEQ1/\-FEF1.
\end{mytheorem}

\begin{proof}
We start with the first part. Let there be $v_1,v_2$, and $r_1,r_2,r_3$. Define the profits and feasibilities as: $p_{11}=1,p_{12}=1,p_{13}=5, p_{21}=1,p_{22}=1,p_{23}=5$ and $f_{11}=1,f_{12}=1,f_{13}=1,f_{21}=0,f_{22}=1,f_{23}=1$. Further, pick $(R_1,R_2)$ with $R_1=\lbrace r_3\rbrace$ and $R_2=\lbrace r_1,r_2\rbrace$. This one is complete. We observe that $p_1(R_1)>p_2(R_2)$/$p_1(R_1)>p_1(R_2)$ holds. We also observe that $p_2(\lbrace r_2\rbrace)=1>0=p_1(R_1\setminus\lbrace r_3\rbrace)$/$p_2(\lbrace r_2\rbrace)=1>0=p_2(R_1\setminus\lbrace r_3\rbrace)$ holds. Hence, $(R_1,R_2)$ is FEQ1/FEF1. However, $(R_1,R_2)$ is not feasible because of $r_1\in R_2$ and $f_{21}=0$. We end with the second part. In the same setting, pick $(R_1\cup R_2,\emptyset)$. This one is feasible complete, but it violates FEQ1/FEF1 because of $p_2(\emptyset)=0<1=p_1(\lbrace r_2\rbrace)<5=p_1(\lbrace r_3\rbrace)$/$p_2(\emptyset)=0<1=p_2(\lbrace r_2\rbrace)<5=p_2(\lbrace r_3\rbrace)$. \QEDB
\end{proof}

By this result, it follows that not every algorithm for computing feasible complete assignments can guarantee FEF1/FEQ1, and not every algorithm for complete FEF1/FEQ1 assignments can guarantee feasibility.

Although the set of feasible complete EQ1/EF1 assignments may be empty in some settings, the set of feasible complete FEQ1/FEF1 assignments is non-empty in each setting. We give algorithms for computing such assignments.

\subsection{FEQ1}

We first prove that feasible complete FEQ1 assignments always exist. More specifically, we give our first new algorithm for computing such assignments in every setting: Algorithm~\ref{alg:feqone}. 

\begin{algorithm}[t]
\caption{A feasible complete FEQ1 assignment with monotone profits.}\label{alg:feqone}
\textbf{Input}: $V$, $R$, $(p_{i})_{n}$, $(f_{ij})_{n\times m}$ \\
\textbf{Output}: a feasible complete FEQ1 assignment
\begin{algorithmic}[1]
\Procedure{Feasible-Min-Max}{}
\State {\sc RemVeh }$\gets V$
\State {\sc RemFeasReq }$\gets \lbrace r_j\in R|\exists v_i\in V:f_{ij}>0\rbrace$
\For{$v_i\in V$}
\State $R_i\gets\emptyset$
\EndFor
\While{{\sc RemFeasReq }$\neq\emptyset$}
\State $v_i\gets\arg\min_{v_k\in \mbox{{\scriptsize\sc RemVeh }}} p_k(R_k)$
\If{$\exists r_h\in \mbox{{\sc RemFeasReq }}:f_{ih}>0$}
\State $r_j\gets \arg\max_{r_h\in \mbox{{\tiny\sc RemFeasReq }}:f_{ih}>0} p_i(R_i\cup\lbrace r_h\rbrace)$
\State $R_i\gets R_i\cup\lbrace r_j\rbrace$
\State {\sc RemFeasReq }$\gets$ {\sc RemFeasReq }$\setminus\lbrace r_j\rbrace$
\Else 
\State {\sc RemVeh }$\gets$ {\sc RemVeh }$\setminus\lbrace v_i\rbrace$
\EndIf
\EndWhile
\State \Return $(R_1,\ldots,R_n)$
\EndProcedure
\end{algorithmic}
\end{algorithm}

At each round, Algorithm~\ref{alg:feqone} selects, among the remaining vehicles, a vehicle whose driver has the least profit over the assigned requests (i.e.\ {\em min}). Then, it checks whether the selected vehicle is feasible for at least one remaining request. If ``yes, it assigns to its driver one of their marginally most profitable feasible remaining requests (i.e.\ {\em max}). If ``no'', it removes the vehicle from the set of remaining vehicles.

\begin{mytheorem}\label{thm:four}
In each setting with vehicle feasibilities and monotone driver profits, Algorithm~\ref{alg:feqone} returns a feasible complete FEQ1 assignment, given that it uses oracle access to the profits.
\end{mytheorem}

\begin{proof} We observe that the returned $(R_1,\ldots,R_n)$ is feasible complete. This is because each request $r_j$, for which some vehicle is feasible, is assigned to some driver $i$ with vehicle $v_i$ that is feasible for $r_j$, i.e.\ $f_{ij}>0$. 

WLOG, we let $(r_1,\ldots, r_s)$ denote the order in which the requests are assigned to drivers. We note that $s\leq m$ holds because the algorithm does not assign requests for which no vehicle is feasible. We also let $h$ denote the iteration number at termination. We note that $h\geq (s+1)$ holds because the algorithm may not assign requests at some iterations. We next let $M(l)=(R_1(l),\ldots,R_n(l))$ denote the partial assignment at the start of iteration $l$. By induction on $l\in\lbrace 1,\ldots, h\rbrace$, we will prove that $M(l)$ is FEQ1. The result will then follow for $l=h$.

In the base case, we let $p=1$. $M(1)=(\emptyset,\ldots,\emptyset)$ is FEQ1 by definition. In the hypothesis, we let $l<h$ and suppose that $M(l)$ is FEQ1. We note that {\sc RemFeasReq} $\neq\emptyset$ holds. In the step case, if no $r_j$ is assigned at iteration $l$ then $M(l+1)=M(l)$. This one is FEQ1 by the hypothesis. For this purpose, we suppose that $r_j$ is the request assigned to $v_i$ at iteration $l$. We note that {\sc RemVeh} $\neq\emptyset$ holds. Let us look at $M(l+1)=(R_1(l),\ldots,R_i(l+1),\ldots,R_n(l))$, where $R_i(l+1)=R_i(l)\cup\lbrace r_j\rbrace$. We will prove in this case that $M(l+1)$ is FEQ1. 

For $v_k,v_l\in V$ with $k\neq i,l\neq i$, we observe that driver $k$ is FEQ1 of driver $l$ and driver $l$ is FEQ1 of driver $k$. These follow because in $M(p+1)$ their profits are the same as in $M(p)$. For this reason, we consider three cases for the drivers of $v_i$ and $v_k\in V\setminus\lbrace v_i\rbrace$. 

{\em Case 1}: Pick $v_i\in$ {\sc RemVeh} and $v_k\in V\setminus\lbrace v_i\rbrace$. By the monotonicity, driver $i$ have in $M(l+1)$ a profit that is not lower than their profit in $M(l)$. Consequently, $p_i(F_{ii}(l+1))=p_i(F_{ii}(l)\cup\lbrace r_j\rbrace)\geq p_i(F_{ii}(l))\geq p_k(F_{ik}(l)\setminus\lbrace r_s\rbrace)=p_k(F_{ik}(l+1)\setminus\lbrace r_s\rbrace)$ holds for some $r_s\in F_{ik}(l+1)$ because of $f_{ij}>0$ and $F_{ik}(l+1)=F_{ik}(l)$. Driver $i$ is FEQ1 of driver $k$.

{\em Case 2}: Pick $v_k\in V\setminus${\sc RemVeh} and $v_i\in$ {\sc RemVeh}. We observe that the driver of $v_k$ has in $M(l+1)$ the same bundle and, therefore, profit as in $M(l)$. Therefore, we conclude $F_{kk}(l+1)=F_{kk}(l)$. We derive $p_k(F_{kk}(l+1))=p_k(F_{kk}(l))\geq p_i(F_{ki}(l)\setminus\lbrace r_s\rbrace)=p_i(F_{ki}(l+1)\setminus\lbrace r_s\rbrace)$ for some $r_s\in F_{ik}(l+1)$ because of $f_{kj}=0$ and $F_{ki}(l+1)=F_{ki}(l)$. Driver $k$ is FEQ1 of driver $i$.

{\em Case 3}: Pick $v_k\in$ {\sc RemVeh}$\setminus\lbrace v_i\rbrace$ and $v_i\in$ {\sc RemVeh}. For the driver of $v_k$, we observe that $p_k(R_{k}(l))\geq p_i(R_{i}(l))$ holds by the fact that $v_i$ is the least bundle profit vehicle within {\sc RemVeh}. Therefore, as the drivers of $v_i$ and $v_k$ are assigned only feasible requests, $p_k(F_{kk}(l))\geq p_i(F_{ii}(l))$ also holds because of $F_{kk}(l)=R_{k}(l)$ and $F_{ii}(l)=R_{i}(l)$. 

As $v_k$ might not be feasible for all requests in $F_{ii}(l)$, we derive $F_{ii}(l)\supseteq F_{ki}(l)$. By monotonicity, it follows that $p_i(F_{ii}(l))\geq p_i(F_{ki}(l))$ holds. Consequently, $p_k(F_{kk}(l))\geq p_i(F_{ki}(l))$ holds as well. 

We note that $F_{kk}(l+1)=F_{kk}(l)$ holds because the driver of $v_k$ has in $M(l+1)$ the same bundle and, therefore, profit as in $M(l)$. If $f_{kj}=0$, we derive $F_{ki}(l+1)=F_{ki}(l)$ and $p_k(F_{kk}(l+1))\geq p_i(F_{ki}(l+1))$. Otherwise, we derive $F_{ki}(l+1)\setminus\lbrace r_j\rbrace=F_{ki}(l)$ and $p_k(F_{kk}(l+1))\geq p_i(F_{ki}(l+1)\setminus\lbrace r_j\rbrace)$. Driver $k$ is FEQ1 of driver $i$. \QEDB
\end{proof}

By comparison, Algorithm~\ref{alg:feqone} may fail to return feasible complete FEF1 assignments. To see this, we refer the reader to Example~\ref{exp:two}. 

\subsection{FEF1}
 
We next prove that feasible complete FEF1 assignments always exist. For computing such assignments, we give a non-trivial trifid modification of the {\em classical} envy graph algorithm from \cite{lipton2004}. In fair division without feasibilities, this algorithm assigns items in an arbitrary order. At each round, it constructs a directed {\em envy graph} that encodes the envy relation between the agents wrt the (incomplete) assignment. Then, it assigns the current item to an {\em unenvied} agent. After that, it eliminates all {\em envy cycles} that may form as a result. This elimination allows the algorithm to ensure that there is an unenvied agent to whom it can assign the next item. In our setting with feasibilities, this algorithm ignores the feasibilities and, for this reason, it may return non-feasible assignments. 

\begin{center}
\fbox{
\begin{minipage}{0.95\columnwidth}
{\bf Modification 1 (Feasible Envy Graph)}: Given (incomplete) assignment $(R_1,\ldots,R_n)$, we define {\em feasible envy graph} as $(V,E)$, where the vertex set $V$ is the set of vehicles and the edge set $E$ is such that $(v_i,v_k)\in E$ iff $p_i(F_{ii})<p_i(F_{ik})$. We define also {\em feasible envy graph projection} wrt $r_j\in R$ as $(V_j,E_j)$, where $V_j=\lbrace v_i\in V|f_{ij}>0\rbrace$ and $(v_i,v_k)\in E_j$ iff $v_i\in V_j,v_k\in V_j,(v_i,v_k)\in E$. In each round of the algorithm, we propose first to use both the envy graph and $(V_j,E_j)$
\end{minipage}
}
\end{center}

\begin{center}
\fbox{
\begin{minipage}{0.95\columnwidth}
{\bf Modification 2 (Unenvied Driver)}: In each round of the algorithm, we propose not to assign the current $r_j$ to a driver that is unenvied wrt the envy graph, but to a driver that is {\em unenvied} wrt $(V_j,E_j)$. 
\end{minipage}
}
\end{center}

\begin{center}
\fbox{
\begin{minipage}{0.95\columnwidth}
{\bf Modification 3 (Feasible Cycle Elimination)}: In each round of the algorithm, we propose not to eliminate the cycles in the envy graph, but the {\em feasible cycles} in $(V,E)$. Thus, if $p_i(F_{ii})<p_i(F_{ik})$ holds in such a cycle then driver $i$ must receive $R_k$ after the cycle is eliminated. In each round of the algorithm, after this step, we also propose to return to the pool of unassigned requests all swapped requests that, as a result of all cycle eliminations, are assigned to vehicles that are non-feasible for them. 
\end{minipage}
}
\end{center}

We are now ready to present our second new algorithm: Algorithm~\ref{alg:fefone}. This algorithm implements modifications 1, 2, and 3. Thus, it returns feasible complete FEF1 assignments in every setting.
 
\begin{algorithm}[h]
\caption{A feasible complete FEF1 assignment with monotone profits.}\label{alg:fefone}
\textbf{Input}: $V$, $R$, $(p_{i})_{n}$, $(f_{ij})_{n\times m}$ \\
\textbf{Output}: a feasible complete FEF1 assignment
\begin{algorithmic}[1]
\Procedure{Feasible-Envy-Graph}{}
\State {\sc RemVeh }$\gets V$
\State {\sc RemFeasReq }$\gets \lbrace r_j\in R|\exists v_i\in V:f_{ij}>0\rbrace$
\For{$v_i\in V$} 
\State $R_i\gets\emptyset$
\EndFor
\While{{\sc RemFeasReq }$\neq\emptyset$}
\State pick some $r_j$ from {\sc RemFeasReq } 
\State compute the feasible envy graph projection $(V_j,E_j)$ \Comment{Modification 1}
\State pick vehicle vi of driver i that is unenvied wrt $(V_j,E_j)$ \Comment{Modification 2}
\State $R_i\gets R_i\cup\lbrace r_j\rbrace$
\State {\sc RemFeasReq }$\gets$ {\sc RemFeasReq }$\setminus\lbrace r_j\rbrace$
\State $(R_1,\ldots,R_n)\gets$ eliminate all feasible cycles in $(V,E)$  \Comment{Modification 3}
\State $\forall v_i\in V: \mbox{\sc RetNonFeasReq}_i\gets \lbrace r_s\in R_i|f_{is}=0\rbrace$  
\State $\forall v_i\in V: R_i\gets R_i\setminus \mbox{\sc RetNonFeasReq}_i$
\State {\sc RemFeasReq }$\gets$ {\sc RemFeasReq }$\textstyle\bigcup_{v_i\in V} \mbox{\sc RetNonFeasReq}_i$
\EndWhile
\State \Return $(R_1,\ldots,R_n)$
\EndProcedure
\end{algorithmic}
\end{algorithm}  
 
\begin{mytheorem}\label{thm:five}
In each setting with vehicle feasibilities and monotone driver profits, Algorithm~\ref{alg:fefone} returns a feasible complete FEF1 assignment, given that it uses oracle access to the profits.
\end{mytheorem}

\begin{proof}
Pick the round of $r_j$. We let $(R_1,\ldots,R_n)$ and $(S_1,\ldots,S_n)$ denote the start-round and end-round assignments. (1) If there are no cycles within the round, the algorithm moves to the next round with one less unassigned request. By monotonicity, as $S_i=R_i\cup\lbrace r_j\rbrace$ and $S_k=R_k$ for each $k\neq i$, $\textstyle\sum_{v_i\in V} p_i(\lbrace r_t\in R_i|f_{it}>0\rbrace)\leq\textstyle\sum_{v_i\in V} p_i(\lbrace r_t\in S_i|f_{it}>0\rbrace)$ holds. (2) Otherwise, it may return some of the assigned requests to the set of unassigned requests. By monotonicity, as cycle eliminations strictly increase the feasible profits of some drivers, $\textstyle\sum_{v_i\in V} p_i(\lbrace r_t\in R_i|f_{it}>0\rbrace)<\textstyle\sum_{v_i\in V} p_i(\lbrace r_t\in S_i|f_{it}>0\rbrace)$ holds. 

These observations are valid for each round. We note that there can be at most $m$ consecutive rounds where (1) holds. After that, either all requests from {\sc RemFeasReq} are assigned or there is a round where (2) holds. We let $U$ denote the maximum value $\textstyle\sum_{v_i\in V} p_i(F_{ii})$ of the feasible utilitarian welfare in any assignment $(R_1,\ldots,R_n)$ and $u$ denote the minimum strictly positive difference $|p_i(X)-p_k(Y)|$ for each $v_i,v_k\in V$ and each $X\subseteq \lbrace r_t\in R|f_{it}>0\rbrace,Y\subseteq\lbrace r_t\in R|f_{kt}>0\rbrace$. After $O(m\frac{U}{u})$ rounds, this welfare reaches $U$ or {\sc RemFeasReq}$=\emptyset$ holds. Then, the algorithm terminates after that many rounds. 

Let us consider the round of $r_j$. Assume that $(R_1,\ldots,R_n)$ is FEF1. Within this round, $r_j$ is assigned to driver $i$. By monotonicity, for $v_i\in V$, $p_i(R_i\cup\lbrace r_j\rbrace)\geq p_i(R_i)$ holds. Therefore, $p_i(F_{ii}\cup\lbrace r_j\rbrace)\geq p_i(F_{ii})$ holds. By assumption, for each $v_k\in V$, it follows $p_i(F_{ii})\geq p_i(F_{ik}\setminus\lbrace r_s\rbrace)$ for some $r_s\in F_{ik}$. Hence, $p_i(F_{ii}\cup\lbrace r_j\rbrace)\geq p_i(F_{ik}\setminus\lbrace r_s\rbrace)$ holds. By the choice of driver $i$, for each $v_k\in V_j\setminus\lbrace v_i\rbrace$ with $f_{kj}>0$, driver $k$ has no feasible envy for $R_i$. Hence, $p_k(F_{ki})-p_k(F_{kk}) \leq 0$ and $p_k(F_{ki})-p_k(F_{kk}) \leq p_k(r_j)$ hold. By assumption, for each $v_k\not\in V_j$ with $f_{kj}=0$, it follows $p_k(\lbrace r_t\in R_i\cup\lbrace r_j\rbrace|f_{kt}>0\rbrace)-p_k(F_{kk})=p_k(F_{ki})-p_k(F_{kk}) \leq p_k(r_s)$ for some $r_s\in \lbrace r_t\in R_i\cup\lbrace r_j\rbrace|f_{kt}>0\rbrace$. Hence, $(R_1,\ldots,R_i\cup\lbrace r_j\rbrace,R_{i+1},\ldots,R_n)$ is FEF1. 

We let $(R_1,\ldots,R_n)$ denote this assignment. After all cycles are eliminated, the bundles $R_1,\ldots,R_i$ are permuted to $R_{j_1},\ldots,R_{j_n}$, respectively. For each $v_k\in V$, the feasible profit of driver $k$ could not have decreased. Therefore, $p_k(\lbrace r_s\in R_{j_k}|f_{ks}>0\rbrace)\geq p_k(F_{kk})$ holds. By assumption, for each $R_{j_h}$ that was previously $R_l$, it follows $p_k(\lbrace r_t\in R_{j_h}|f_{kt}>0\rbrace)-p_k(\lbrace r_t\in R_{j_k}|f_{kt}>0\rbrace)\leq p_k(F_{kl})-p_k(F_{kk})\leq p_k(r_s)$ for some $r_s\in \lbrace r_t\in R_{j_h}|f_{kt}>0\rbrace$. Therefore, $(R_{j_1},\ldots,R_{j_n})$  is FEF1. 

We let $(R_1,\ldots,R_n)$ denote this assignment. After all non-feasible requests are returned, $R_1,\ldots,R_n$ are changed to $S_1=R_{1}\setminus\mbox{\sc RetNonFeasReq}_1,\ldots,S_n=R_n\setminus\mbox{\sc RetNonFeasReq}_n$, respectively. For each $v_k\in V$, we derive $p_k(S_k)=p_k(F_{kk})$. By assumption, for each $R_l$, it follows $p_k(F_{kk})\geq p_k(F_{kl}\setminus\lbrace r_s\rbrace)$ for some $r_s\in F_{kl}$. But, $R_l$ is now $S_l$. If $r_s\not\in\mbox{\sc RetNonFeasReq}_l$, then $F_{kl}\setminus\lbrace r_s\rbrace\supseteq (F_{kl}\setminus\mbox{\sc RetNonFeasReq}_l)\setminus\lbrace r_s\rbrace$ holds. Otherwise, $F_{kl}\setminus\lbrace r_s\rbrace\supseteq F_{kl}\setminus\mbox{\sc RetNonFeasReq}_l$ holds. By monotonicity, we derive $p_k(F_{kk}\setminus\lbrace r_s\rbrace)\geq p_k((F_{kl}\setminus\mbox{\sc RetNonFeasReq}_l)\setminus\lbrace r_t\rbrace)$ for some $r_t\in F_{kl}\setminus\mbox{\sc RetNonFeasReq}_l$. Consequently, $(S_1,\ldots,S_n)$ is FEF1. As the algorithm terminates, all requests from {\sc RemFeasReq} are assigned. The returned assignment is FEF1. 

Finally, let us consider the round of $r_j$. Assume that $(R_1,\ldots,R_n)$ is feasible. Eliminating all cycles may permute the bundles of drivers and thus ruin the feasibility. However, returning all non-feasible requests to {\sc RemFeasReq} restores this feasibility. By assumption, $(S_1,\ldots,S_n)$ is guaranteed to be feasible. As the algorithm terminates, all requests from {\sc RemFeasReq} are assigned. The returned assignment is feasible complete. \QEDB
\end{proof}

By comparison, Algorithm~\ref{alg:fefone} may fail to return feasible complete FEQ1 assignments. To see this, we refer the reader to Example~\ref{exp:two}. 

\section{Future directions}\label{sec:fut}

In the future, we will study the interaction between FEQ1, FEF1, and an economic concept such as Pareto optimality, as well as manipulations of algorithms for FEQ1 and FEF1. We will also consider VRP settings with time windows and ordering constraints, and study the impact fairness has on the total delay in such settings. Non-monotone preferences are also left for future work. Finally, we will measure through simulations the quality of routing in fair assignments.

\section{Acknowledgements}\label{sec:ack}

Martin Damyanov Aleksandrov was supported by the DFG Individual Research Grant on ``Fairness and Efficiency in Emerging Vehicle Routing Problems'' (497791398).

\bibliographystyle{splncs04}
\bibliography{submission}

\end{document}